\pdfoutput=1 
\documentclass[11pt]{article}
\usepackage[letterpaper, margin=1in]{geometry}
\usepackage{cite}
\usepackage{graphicx}
\usepackage{amsmath}
\usepackage{amsthm}

\usepackage{url}
\usepackage{color}
\usepackage{dsfont}
\usepackage{setspace}

\usepackage[affil-it]{authblk}

\newtheorem{theorem}{Theorem}[section]
\newtheorem{lemma}[theorem]{Lemma}
\newtheorem{definition}{Definition}

\newcommand{\argmax}[1]{\underset{#1}{\operatorname{argmax}}}

\usepackage{caption}

\usepackage{color}
\usepackage[pagebackref,colorlinks,citecolor=blue,linkcolor=blue,linktocpage=true]{hyperref}

\newcommand{\A}{\mathcal{A}}

\newcommand{\mechname}{{\sc align-trust}}
\newcommand{\bks}{\textsc{bks}}
\newcommand{\vmm}{\textsc{vmm}}
\newcommand{\fq}{\textsc{fq}}
\newcommand{\vcg}{\textsc{vcg}}
\newcommand{\spq}{\textsc{spq}}
\newcommand{\fifo}{\textsc{fifo}}

\makeatletter

\begin{document}

\date{}
\title{Truthful Prioritization Schemes for Spectrum Sharing}


\author{Victor Shnayder, David C. Parkes%
  \thanks{\texttt{\{shnayder,parkes\}@eecs.harvard.edu}}}
\affil{School of Engineering and Applied Sciences\\
 Harvard University}

\author{Vikas Kawadia%
  \thanks{\texttt{vkawadia@bbn.com}}}
\affil{Raytheon BBN Technologies}

\author{Jeremy Hoon%
  \thanks{\texttt{jhoon@stripe.com}}}
\affil{Stripe}

\maketitle

\begin{abstract}
We design a protocol for dynamic prioritization of data on shared routers such
as untethered 3G/4G devices. The mechanism prioritizes bandwidth in favor of
users with the highest value, and is incentive compatible, so that users can
simply report their true values for network access. A revenue pooling mechanism
also aligns incentives for sellers, so that they will choose to use
prioritization methods that retain the incentive properties on the buy-side. In
this way, the design allows for an open architecture. In addition to revenue
pooling, the technical contribution is to identify a class of stochastic demand
models and a prioritization scheme that provides allocation monotonicity.
Simulation results confirm efficiency gains from dynamic prioritization relative
to prior methods, as well as the effectiveness of revenue pooling.
\end{abstract}

\clearpage

\tableofcontents

\clearpage

\section{Introduction}

Many mobile broadband data users overpay for  data plans, buying a data plan
with sufficient monthly quota for their maximum needs even though the average
consumption is only about 15\% of the monthly quota~\cite{Cisco2011}. Analysis
of cell phone data has also shown \emph{quota
dynamics}~\cite{andrewsunderstanding} in users, i.e., sensitivity  to quota
balance and time to the end of the quota period.  One way to address these
inefficiencies is to allow surplus bandwidth to be shared. In this paper we
propose an auction-based protocol for such bandwidth-sharing. Sellers have
wireless broadband (3G/4G) devices with a WiFi like WLAN radio, and can
share network access through tethering apps that allow them to act as routers.
Buyers have WiFi devices and pay sellers to relay their data.

Rather than finding optimal static allocations to users, our focus is on dynamic
prioritization of access to bandwidth.  Dynamic prioritization is more efficient
because it can handle temporal heterogeneity in user demand.  We design an
auction protocol to prioritize network access in favor of those with highest
value. In ensuring simplicity for users, we seek an incentive compatible design,
so that truthful reporting of value and straightforward use of the shared
bandwidth (no delaying of traffic, no padding of traffic) is optimal for a user.
Incentive compatibility is also useful in avoiding ``churn'' and system overhead
that can occur if users can benefit by adapting reported values given reports of
others. Similar arguments have been made in the context of sponsored search
markets~\cite{Edelman2007a}.

Our positive results are stated for users with linear (per-byte) valuation
functions. Give this, incentive compatibility requires that the cumulative
quantity of network resources consumed by a user is non-decreasing in bid value,
a property referred to as {\em monotonicity}. Our main technical contribution is
to establish conditions on user demand models for which a strict priority-queue
approach to resource access control satisfies monotonicity. An important
property satisfied by the demand models is that a user's cumulative consumption
of network bandwidth over a user session weakly increases as a function of her
total consumption up to any intermediate time. An additional technical challenge
that we address is to ensure that users cannot benefit through delayed use of
allocated network resources, or by introducing fake traffic to increase demand.

Payments are computed by adopting an approach due to Babaioff, Kleinberg, and
Slivkins~\cite{Babaioff:2009p602} (\bks{}). This involves adding an additional, random
perturbation to bids.\footnote{A side effect is that the scheme does not reduce
to fair share in the case of $n$ identical buyers. Rather, the scheme randomly
perturbs each bid before ranking them for the purpose of determining network
prioritization.}  In a typical auction setting, the auctioneer knows how an
allocation of resources would have changed given different bids. This
counterfactual information is essential to standard payment schemes but
unavailable here. It would require knowledge of how much bandwidth a user would
consume under a different priority, but this requires knowledge  that  the
network infrastructure does not have about a user's demand model. The only
information available for the purpose of determining payments is the actual
realization of consumption based on the actual assignment of priority. The
introduction of random perturbation by the \bks{} scheme avoids the need for
this counterfactual information. 

We also extend our design so that it embraces open network
architectures.  In particular, we are robust to router
devices that can install alternative routing software,
for example to change prioritization schemes, or otherwise tamper with
methods to compute payments. Thus, we seek to align incentives
on the sell side as well as the buy side of the market.  Our solution
adopts lightweight cryptography and revenue pooling across sellers,
and has the effect of aligning the interests of sellers with adopting
routing policies that maximize total buyer value. This is sufficient
for monotonicity of user allocations, which is in turn sufficient for
buyer truthfulness. Revenue pooling makes the market design appear
simultaneously as a first-price market for sellers and a second-price
market for buyers. For sellers, this means they prefer to maximize the
(bid) value of the allocation. For buyers, this retains the
second-price-like \bks{} scheme, and thus incentive compatibility.

A simulation study confirms that the mechanism achieves arbitrarily
close approximations to full allocative efficiency for simple demand
models (allocating the shared resource to those who value it the
most), shows that reserve prices can increase efficiency of
non-prioritized routing methods, demonstrates a scenario where sellers
have practical efficiency-improving deviations, and examines the
distributional effects of revenue pooling. We have also prototyped our scheme
using the traffic control module in the Linux kernel.

\subsection{Related work}

Sen et al.~\cite{Sen:DataPricingSurvey} provide an exhaustive survey of the
large literature on smart data pricing. We focus on describing work related to the specific
techniques that we use.

Babaioff et al.~\cite{Babaioff:2009:CTM} and Devanur et
al.~\cite{Devanur:2009:PTP} show that the unavailability of counterfactuals
impedes the design of truthful mechanisms in the context of multi-armed bandit
problems. We employ the \bks{} scheme in a new application domain.  The effect
is that the payment scheme that we adopt accounts for varying user demand,
providing an unbiased estimate of the cost imposed on other buyers by the
prioritization associated with a buyer's device.  In contrast, an earlier scheme
due to Varian and Mackie-Mason (VMM)~\cite{Varian:1994p598} adopts a myopic,
per-packet viewpoint on the cost imposed on other users and is not incentive
compatible when buyers have adaptive demand patterns. 

Other approaches from dynamic mechanism design are unsuitable, either because
they rely on counterfactual information~\cite{hajiaghayi05} or rely on a
probabilistic demand model~\cite{parkes03c}.  VMM's work inspired other
approaches, such as the progressive second-price auction~\cite{Lazar:1998p588},
and some follow-on work~\cite{Maille:2003p586}.  As with VMM, they also do not
achieve incentive alignment in dynamic settings.  They also assume the existence
of a trusted router.

Godfrey et al.~\cite{Godfrey2010a} study incentive compatibility of
congestion control mechanisms in networks.  Our technical analysis of
the separation between bidders is similar to the separation between
flows in this earlier model, but the domain of study and main results
are otherwise incomparable. We consider user bidding and payments,
whereas they analyze the effect of manipulations in forwarding
strategies on network-wide congestion.

In a different domain, Shneidman and Parkes~\cite{Shneidman2004} study the
problem of faithful network protocols for BGP routing, where the algorithms
adopted by network users must themselves form part of an equilibrium. In this
sense, our revenue pooling scheme attains faithfulness with respect seller
routing algorithms.

Also related is a large literature on the use of cryptographic solutions to
provide trustworthy auctions; e.g.~\cite{parkes07}.  However, these solutions
incur too much overhead in the context of dynamic bandwidth prioritization.
Porter and Shoham~\cite{porter05} provide an analysis of how the presence of
cheating provides a second-price auction with first-price semantics. Our revenue
pooling scheme, used for aligning trust does the opposite: we give a
second-price auction first-price semantics from the perspective of the
auctioneer (or seller, in our model), thereby mitigating incentives for
manipulation; the scheme borrows ideas from a random-sampling approach used in a
very different domain, that of the design of revenue-optimal digital good
auctions~\cite{Goldberg01}. 

\section{Model}

In this section we describe our model, also shown in Fig.~\ref{fig:bandex_model}.
\begin{figure}[t]
\begin{center}
\includegraphics[width=0.65\linewidth]{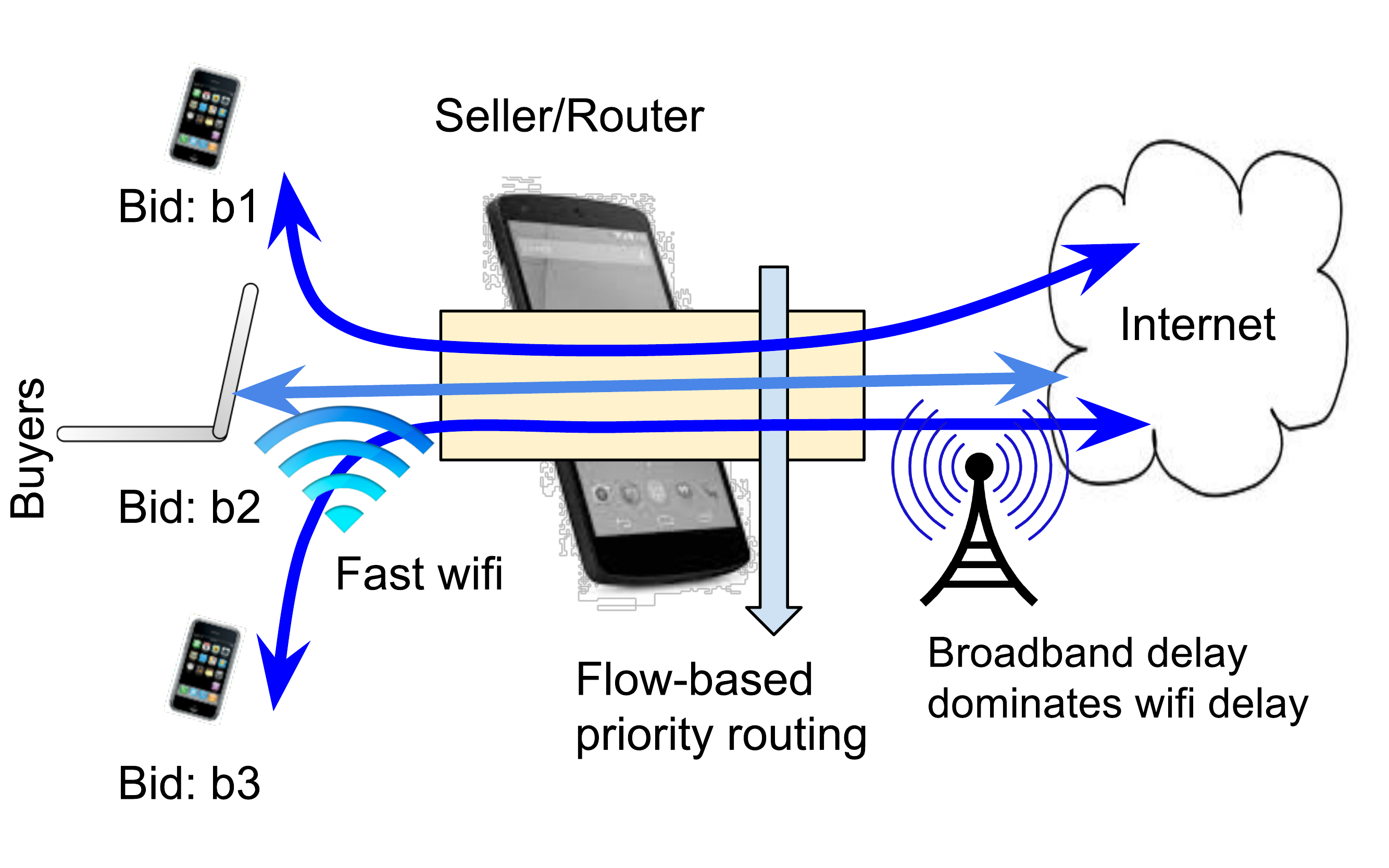}
\caption{Illustration of our model.}
\label{fig:bandex_model} 
\end{center}
\end{figure}

{\bf Time and value:} Time is modeled in discrete epochs. Each buyer
(agent) $i$ has an allocation period $[\alpha_i, \beta_i]$, with
$\alpha_i \ge 0$ and $\beta_i \ge \alpha_i$. This is the period of
time over which the buyer associates value with receiving bandwidth.
Each buyer has constant value $v_i\ge 0$ per byte (a linear
valuation), and bids once before sending data.

{\bf Routing:} We adopt a flow model of traffic, with a single flow modeling
both upstream and downstream traffic for each buyer.  The router's link to the
Internet is assumed to have fixed capacity $c>0$.  Each buyer has bid $b_i \ge
0$, and receives a {\em priority} based on this reported value.

We consider three schemes for network prioritization:

1. First-in First-out (\fifo{}) routes data in a first-in, first-out order as it
arrives at the router. Since we do not model the queue of a router directly, we
model \fifo{} as a within-epoch allocation proportional to within-epoch demand.

2. Fair queueing (\fq{}) allocates bandwidth to all buyers evenly, dividing any
unused capacity recursively: in each period, each of the $n$ buyers is
guaranteed to be able to consume $c/n$.  Any unused capacity is divided evenly
among buyers with further demand, and this repeats until everyone is satisfied
or capacity is exhausted.

3. Strict Priority Queueing (\spq{}) first allocates capacity to the buyer with
the highest bid in each epoch, up to $c$ or the demand of the buyer. Then,
capacity is allocated to the buyer with the second highest bid, and so forth, as
long as capacity is available. Ties are broken at random. 

\label{sec:policies}
\smallskip

Let $c_{i,t}$ denote the network capacity available to buyer $i$ in epoch $t$
(this depends on her bid, as well as the bid and network usage of others.)

{\bf Transport: }  We assume that the local network between router and buyers is
much faster than the router's link to the Internet and neglect delays there.  

{\bf Demand: } Let $x_{i,t}$ (bytes) denote the  cumulative amount of traffic
(upload or download) associated with buyer $i$ up to and including time $t$.
This may be smaller than the total capacity available to buyer $i$ up to and
including time $t$ because it depends on the buyer's demand.

Let $D_i(t,x)$ (bytes) denote the demand of buyer $i$ during epoch $t$, where
$x$ is the cumulative amount of traffic (upload or download) used so far; i.e.,
$x=x_{i,t-1}$. Equivalently, this represents the maximum amount of network
capacity that buyer $i$ wants to consume during the epoch. $D_i(t,x)$ is a
random variable, and we adopt $d_i(t,x)$ to denote a specific realization. We
consider demand models $D_i$ that satisfy the following condition:

\footnote{Conceptually, we can think about the demand  for all
    $(t,x)$ values as being realized by ``nature'' when a buyer arrives,
capturing the type of the buyer (even if it is unknown to the buyer.)}

\begin{definition}
\label{def:natural}
A demand model $D_i(t,x)$ is \emph{natural}  if all realizations $d_i(t,x)$ 
satisfy for all $t$, for all $x \ge x'$ and all $c \ge 0$:
$$x + \min(c,d_i(t, x)) \ge x' + \min(c,d_i(t, x')).$$
\end{definition}

A demand model is natural when 
getting more capacity earlier might
increase demand in the future, and should not decrease future demand
below the total amount given lower capacity. 
\label{sec:natural-examples}
In presenting some examples of demand models that are natural in this
sense, we focus on deterministic examples.  Randomized models can be
created through random perturbations to the parameters of the example,
and also through randomization over the different kinds of models.

\emph{Constant demand:}
A buyer $i$ who simply wants to send at some constant rate: $d_i(t,x) = k$. 

\emph{Time-varying demand:} A buyer $i$ with an arbitrary time-dependent
generation process $g(t)$ that generates data that must be sent immediately to
be useful: $d_i(t,x) = g(t)$.

\emph{Buffered demand:} A buyer $i$ with an arbitrary time-dependent generation
process $g(t)$ that generates data, and wants to send as much of that as
possible, buffering demand until it is allocated: $d_i(t,x) = \sum_{p \le t}
g(p) - x$.

\emph{Impatient buyer:}
\label{example:impatient}
A buyer $i$ who sends until time $p$, and then goes away if some minimum amount of service $m$ has not been met:
\[d_i(t,x) = \begin{cases}
k \text{ if } t \le p\\
k \text{ if } t > p \text{ and } x > m\\
0 \text{ else }
\end{cases} \]
This can be generalized to multiple thresholds, reduction to a lower, but non-zero rate, etc.

\emph{Increasing demand (in rate)}
A buyer that has $d_i(t,x) = g(x/t)$, with $g$ weakly increasing.

\emph{Increasing demand (in total)}
A buyer that has $d_i(t,x) = g(x)$, with $g$ weakly increasing.
\medskip

This impatient buyer model illustrates that the requirement of "natural"  demand
functions allow those users to be modeled whose future value for network access
falls if they don't receive enough bandwidth.

These demand models capture many realistic types of demand, but some demand functions
do not fit our model.  In particular, we do not allow demand to depend on recent
usage, precluding models like ``if I have not 
been able to use the network
in the last 30 seconds, give up''.

\section{Prioritization and Payments}
\label{sec:mechanisms}

Having introduced the basic elements of our model, we now describe
some variations on payment schemes.

\subsection{Fixed price}

A baseline comparison is provided by a fixed price payment scheme.
This can be used together with the \fifo{} and \fq{} routing policies,
where only the users who bid above the fixed price are considered by
the prioritization schemes.

\subsection{VMM mechanism}
\label{sec:vmm}

The \vmm{} mechanism provides a second comparison point. This mechanism uses the
\spq{} routing policy. For payments, the original paper by Varian and
Mackie-Mason used a per-packet model, and charged the owner of each forwarded
packet the \emph{immediate externality} imposed by the packet.  This is the
value of the highest value packet that was dropped while the forwarded packet
was in the router's queue.  

Since we do not model the router queue directly, we adapt the idea
behind the payment mechanism to our flow-based model by charging the
immediate externality imposed by a buyer's flow in each period,
computed using the standard Vickrey-Clarke-Groves (\vcg{}) mechanism,
charging $V_{-i} - V^*_{-i}$ for each buyer $i$ in each period, where
$V_{-i}$ is the total (reported) value to all the other buyers in the
(counterfactual) optimal allocation when $i$ is removed, and
$V^*_{-i}$ is the total (reported) value to all the other buyers in
the real optimal allocation with $i$ included.

\vmm{} is a natural mechanism, and has inspired a lot of work on
bandwidth pricing. However, the \vmm{} approach is not incentive
compatible in a setting with variable demand. Consider the following:
\smallskip

{\bf Example 1.} \emph{Suppose that the router link capacity is 1 packet per
second, and there is one buyer with value \$3 per packet who wants to send one
packet every second and a second buyer with value \$2 per packet who wants to
send a single packet, and keeps trying until she succeeds.  Truthful bidding by
buyer 1 will result in the second buyer's packet being dropped every period, and
a charge of \$2 for each packet, representing the per-epoch externality on buyer
2.  On the other hand, a bid of less than \$2 would ``flush" the packet of buyer
2 and then allow buyer 1 to send for the remaining epochs with payment \$0. The
tradeoff is to reduce the amount of data forwarded by 1 packet in return for a
significant reduction in total payment.}

\subsection{The BKS Mechanism}
 \label{sec:bks}

In the example above, \vmm{} over-estimates the externality because it does not
have access to the information that buyer 2 only has a single packet to send.
This information is not available, since we assume that demand models are not
described by users or known by the prioritization scheme.  In fact, even the
user themselves may not know their future demand before it is realized.

In this domain, the \bks{} mechanism adopts the \spq{} routing policy.  Payments
are determined following a self-sampling approach, where a randomized
perturbation to bids obviates the need for counterfactual information. The idea
is to obtain an estimate of the network resources that a user would have
consumed at some lower bid as a side-effect of the randomization.

We adapt the scheme to also allow the seller to employ a {\em reserve price}
$r\geq 0$, which is the minimal per-byte price a seller will accept.\footnote{To
  support the reserve price, we use the \emph{h-canonical self-resampling
  procedure} described by \bks{}, with $h(z,b) = r + z (b-r)$, which has
  distribution function $F_h(a,b) = (a-r)/(b-r)$.  In Section 3.4, the \bks{}
  paper claims that $F_h = F_0$ for all $h$, where $F_0$ is the distribution
  function for the canonical resampling procedure, but $F_0$ doesn't satisfy
  their condition on $F_h$: $h(F_h(a,b), b) = a\quad \text{for all } a, b\in I,
a < b$.  $F_h(a,b) = (a-r)/(b-r)$ does satisfy this condition.}  The scheme is
parameterized by $\mu\in (0,1)$, which governs the probability of introducing a
random perturbation into bids. 

The scheme is most easily explained in terms of an \emph{allocation rule} $\A$.
For us, this encapsulates the combined effect of realized demand (including
possible strategic effects by users delaying demand or creating fake demand),
the routing policy, the total capacity constraint on the router, and user bids.
Taken together, these elements define the total amount of traffic consumed by
each user (equivalently, the realized allocation).

The allocation $\A(\tilde{b},d)$ generated by the rule is \emph {ex
  ante} uncertain---it depends on realized demand
$d=(d_1,\ldots,d_n)$, where $d_i$ denotes the realized demand by $i$
in each epoch, and is applied to randomly perturbed bids $\tilde{b}$.
The \bks{} payment scheme works as follows:

\begin{definition}[\bks{}]
{\em Given an allocation rule $\A$ and a parameter $\mu \in (0,1)$, 
the \bks{} procedure in our setting is:
\begin{enumerate}
\item Upon arrival, each bidder $i$ submits a bid $b_i\geq r$. 
\item The mechanism computes transformed bids $\tilde{b}_i$ $\forall$ $i$:  
  \begin{enumerate}
  \item With probability $1-\mu$, $\tilde{b}_i = b_i$
  \item \label{item:resample} Else, compute a reduced bid: pick $\gamma \in [0,1]$ uniformly at random, and set
  $\tilde{b}_i = r + (b_i - r) \cdot \gamma^{1/(1-\mu)}$
  \end{enumerate}
\item For all epochs, the router uses the allocation rule $\A(d,\tilde{b})$ applied to the transformed bids $\tilde{b}$ of the active users and given
realized demand $d=(d_1,\ldots,d_n)$.
\item Given the amount of network
data, $x_i\geq 0$, associated
with each user $i$, (as realized by
the allocation rule $\A$ applied to transformed bids), 
collect payment from user $i$ as:
  \begin{enumerate}
  \item Collect $b_i x_i$.
  \item \label{item:rebate} If $\tilde{b}_i < b_i$, give a \emph{rebate} $R_i = \frac{1}{\mu}(x_i(b_i - r))$. Otherwise, $R_i = 0$.
  \end{enumerate}
\end{enumerate}}
\end{definition}

Bids are perturbed, used for prioritized routing, the total number of
bytes associated with a user is observed, and payments are made
through a randomized adjustment via the rebate in step 4. Each user's
rebate can be determined at the end of her allocation period,
allowing the user's payment to be computed while other users are still
active.  
\begin{definition}
\emph{A mechanism is \emph{truthful-in-expectation} if a risk-neutral buyer
maximizes expected utility by bidding truthfully, 
whatever the bids of others,
where the expectation is taken with respect to random coin flips of
the mechanism.}
\end{definition}

An essential property for truthfulness is the {\em ex post
  monotonicity} of the allocation rule. 

Let $\A_i(b',d)$ denote
the allocation to buyer $i$ given some bid vector $b'$ and demand
vector $d$. (We write $b'$ to denote a generic bid vector and avoid
confusion with the $b$ submitted as input to \bks{}.)
An allocation rule is {\em ex post monotone} if
$\A_i(b'',d)\geq \A_i(b',d)$, for all bid vectors $b'$, all
demand vectors $d$, and all
$b''=(b'_1,\ldots,b_i'',\ldots,b'_n)$ such that $b''_i>b'_i$. This is {\em
  ex post} in the sense that whatever the bids and whatever the demand,
a buyer's total traffic
consumption weakly increases with her bid.
We use the following result:
\begin{theorem}
\label{thm:bks-truthfulness}
Applying the \bks{} procedure with probability of
perturbation $\mu$ to an allocation rule that is {\em ex post} monotone 
results in a truthful-in-expectation mechanism. 
\end{theorem}

The proof in \cite{Babaioff:2009p602} shows that the scheme obtains an unbiased
sample of an integral that defines the payment rule in the canonical
approach of incentive-compatible mechanism design~\cite{Myerson81}.
The allocation is the same as in
the original allocation rule with probability at least $1-n\mu$, where
$n$ is the number of buyers.
\smallskip

{\bf Example 2.} {\em 
Let's revisit the earlier example of manipulation in the \vmm{} scheme. Under
\bks{}, when the first buyer's bid is not resampled, she pays \$3 per packet,
and has some total allocation $k$.  When her bid is resampled, the first buyer
will have allocation $k$ or $k-1$, depending on whether the resampled bid was
below \$2.  This will result in a large rebate, and in expectation, the first
buyer's total payment will be essentially \$2, with the exact value depending on
$\mu$.}

\section{Buyer incentives}
\label{sec:buyer-truthfulness}

To establish truthfulness of the \bks{} mechanism in our setting we need to show
that \spq{} combined with natural demand models implies that the allocation rule
is {\em ex post} monotone.  In our setting, the allocation rule is the process
that determines the total network traffic used by each buyer over her allocation
period. 

We first show that \spq{} provides an isolation property on buyers, and then analyze the monotonicity of our allocation rule.  
\begin{lemma}[isolation]
\label{lem:31}
For any $t$, capacity $c_{i,t}$ under \spq{} is independent of how capacity is used by buyer $i$ in time $< t$, and weakly increases with bid $b_i$.
\end{lemma}
\begin{proof}
Consider $n$ buyers, and order them by bids $b_1\ge b_2\ge
\ldots \ge b_n$ with ties broken at random. Let's first consider the
claim that capacity $c_{i,t}$ in epoch $t$ is independent of how the
capacity is used by buyer $i$ in earlier epochs.  Proceed by strong
induction.  For buyer 1, then this is immediate since the buyer always
gets to use full the capacity of the channel.  For buyer $i>1$, given
the induction hypothesis for buyers $< i$, buyer $i$ cannot affect
demand or allocation to higher priority buyers and is not affected by
the use of the capacity by lower priority buyers. The result is that
buyer $i$ gets all capacity that is unused by the higher
priority buyers. The capacity $c_{i,t}$ is weakly increasing with bid
value because of \spq{} routing. In particular, if buyer $i>1$ gets
a higher priority $i'<i$ then her capacity in period $t$ will be the capacity unused by all buyers $i''<i'$, whereas previously its
capacity was that unused by all buyers $i''<i$.
\end{proof}

Given this lemma, we can now focus on an arbitrary buyer $i$, and simplify
notation by omitting the subscript $i$. In particular, we use $d_t$ to denote
the realized demand of the buyer in period $t$ (keeping the dependence on total
network capacity used so far silent), and $c_t$ to denote the buyer's realized
capacity. Both $d_t$ and $c_t$ are realizations of a random process (the latter
due to its dependence on the demand models of other buyers.)
We first establish monotonicity and truthfulness properties of the mechanism
under the assumption that the buyer is greedy in her use of the channel.
\begin{definition} Given realized capacity $c_t$ and demand $d_t$ in epoch $t$,
  a buyer using the \emph{greedy} policy sends  $\min(c_t, d_t)$ to the router
  in epoch $t$.  \end{definition}

The greedy policy stipulates that the buyer makes her 
demand $d_t$ available to the router (up to capacity $c_t$),
and neither pads the demand with fake traffic nor hides demand
by introducing a delay.
\begin{lemma}[monotonicity]
\label{lemma:monotonicity} If buyer $i$'s demand model is natural and routing is done using \spq{}, then fixing bids of other buyers and realized 
demand, buyer $i$'s allocation up to and including any epoch $t$ under
the greedy policy
is {\em ex post} monotone in bid value $b_i$.
\end{lemma}
\begin{proof}
The proof uses the following lemma.

\begin{lemma}
\label{lemma:maxforward}
If routing is done using \spq{} and the demand function is natural, then for any $t$, bid values
and realized demand, following
the greedy policy in every period up to and including $t$ maximizes the total 
amount of network capacity used by 
the buyer, $x_t$, up to and including period $t$.
\end{lemma}
\begin{proof}
We proceed by induction on $t$. By fixing bid values and fixing realized
demand we have isolation. The base case is simple: if $t=1$ then 
using less than $\min(c_1, d_1)$ for any $c_1, d_1$ does not maximize $x_1$.

Consider $t > 1$.  Suppose for contradiction that there is some natural demand function $d$ and realization of capacity $c$ such that using $\min(c_p, d_p)$ in every epoch $p \le t$ does not maximize $x_t$.  Let $f$ be a  sequence of usage amounts that does maximize $x_t$.

If $f_t < \min(c_t, d_t)$, setting $f_t = \min(c_t, d_t)$ increases the allocation in epoch $t$ 
and thus the total amount consumed, contradicting the fact that $f$ maximizes $x_t$.

If $f_t = \min(c_t, d_t)$, then there exists some $p < t$ such that $f_p < \min(c_p, d_p)$.  Let $x_{t-1}$ be the total amount used up to and including 
epoch $t-1$ under $f$.
Consider a  sequence $f'$ of usage amounts
that is greedy for all epochs $\leq t$, and let $x'_{t-1}$ be the total used up
to and including $t-1$ under $f'$.  By the induction hypothesis, $f'$ maximizes the amount used up to and including $t-1$, so $x'_{t-1} \ge x_{t-1}$.  Because the demand function $d$ is natural,
\[x'_t = x'_{t-1} + d(t, x'_{t-1}) \ge x_{t-1} + d(t, x_{t-1}) = x_t,\]
so $f'$ is indeed a maximizing sequence, establishing a contradiction.
 \end{proof}

Using this lemma, we obtain Lemma~\ref{lemma:monotonicity}. Suppose otherwise.
Then it would be the case that following the
greedy policy and using the 
 capacity $c_{p}$ in every
epoch $p\le t$ leads to a smaller allocation by epoch $t$ than 
using the greedy policy and using the
 capacity $c'_{p}$ where $c'_{p} \le c_{p}$ for all $p\le t$ (and $c_t$ is
the capacity under a higher bid $b_i>b_i'$, with this relationship between
$c'_t$ and $c_t$ following from Lemma~\ref{lem:31}). But this is a contradiction
with Lemma~\ref{lemma:maxforward}.
\end{proof}

\begin{lemma}
If the buyer is restricted to the greedy policy, her demand is natural, and routing is done using \spq{} then the \bks{} mechanism is truthful and
the dominant strategy is to submit a truthful bid $b_i=v_i$.
\end{lemma}
\begin{proof}
Follows directly from Lemma~\ref{lemma:monotonicity} and the properties of \bks{} in Theorem~\ref{thm:bks-truthfulness}.
\end{proof}

But we are also want to show
that users cannot benefit
by delaying network usage, or padding traffic
with fake packets.
\begin{theorem}
Given \spq{} routing and natural demand, the dominant strategy of a user in the \bks{} mechanism is to bid its
true value and follow the greedy policy.
\end{theorem}
\begin{proof}
We first consider padding with fake traffic in some epochs. In
particular, consider the final epoch $t$ in which the buyer does this,
and for it to matter assume $d_t<c_t$. Fix any bid value $b_i$
(perhaps untruthful). We establish that this weakly increases the
buyer's expected payment without providing additional value (since it
is fake traffic.) The expected payment in \bks{} is equal
(through the self-sampling approach) to the
Myerson payment, and thus by the Myerson~\cite{Myerson81}
rule, we require
\begin{align*}
b_ix'_i(b_i)-\int_{0}^{b_i}wx_i'(w)dw\geq b_ix_i(b_i)-\int_{0}^{b_i}wx_i(w)dw,
\end{align*}
where $x'_i(w)$ is the total allocation given bid $w$ with the
padding in epoch $t$ and $x_i(w)$ is the total allocation given
bid $w$ without padding in epoch $t$. 
This inequality
holds since  $x'_i(b_i)-x_i(b_i)\geq x'_i(w)-x_i(w)$ for all $w\leq b_i$,
since for lower values the buyer may at some point receive a 
lower priority at $t$ and thus be able to send less incremental
traffic through padding, from which we have,
\begin{align*}
b_i(x'_i(b_i)-x_i(b_i))&\geq \int_{0}^{b_i}w\left(x'_i(w)-x_i(w)\right)dw.
\end{align*}

Second, consider holding back demand in some epochs such that not all the
capacity is utilized.  Now that padding of demand has been precluded, it follows
from monotonicity with respect to bid value, and the truthfulness of BKS, that
it is a weakly dominant strategy for a user to bid her true per-byte value.
Moreover, since the expected payment in BKS is equal to the payments in the
Myerson auction, then $v_i x^* − z^* \geq v_i x' − z'$ for all alternate
allocations $x' \leq x^*$ obtained through demand reduction, where $x^*$ is the
allocation achieved under the greedy policy, and $z^*$ and $z'$ the expected
payment in BKS at allocation $x^*$ and $x'$ respectively.
From this, it follows immediately that the buyer cannot benefit by holding back
demand.  \end{proof} 

On this basis, we conclude that the \bks{} mechanism has the ``truthful
in expectation'' property in that it supports both truthful bidding and
also straightforward revelation of demand (and use of the capacity
made available through \spq{} applied to perturbed bids).

\section{Sell-side incentives}
\label{sec:sell-side}

We have assumed so far that the seller and router device can be trusted; e.g.,
to follow \spq{} and \bks{} pricing.  However, a self-interested seller can do
better by deviating from the scheme as currently described. This is a concern
because it would also lead buy-side incentives to unravel.
For example, consider the following profitable seller manipulations, and the
failure of simple fixes: \smallskip

(1) The simplest profitable manipulation for sellers is to simply never give
rebates, justifying this by claiming that the randomized bid resampling in
\bks{} just happened to work out that way.  A natural fix to this is to insist
that the seller send the signed bids to a trusted central server for resampling.
The central server would now know whether a buyer should pay her bid or get a
rebate for a particular allocation, and since the server handles the accounting,
it could appropriately credit or debit the buyer's account when it learned the
final allocation.

(2) If the seller cannot tamper with the bid resampling process, it
can still manipulate by reducing the allocation to resampled bidders,
since they will be getting rebates and giving them service will reduce
the seller's revenue.  A partial fix to this is to note that to
implement \spq{}, the seller only needs to know the priority order of
bidders, and not their bids, so bids can be encrypted when sent to the
central server, which would report just the order to the router,
keeping the seller ignorant of the bids.

(3) Unfortunately, the seller can still infer enough from the order to
manipulate, even in a single allocation period.  We omit the math
here, but it can be shown that the expected revenue from the lowest
priority user is negative, so the seller would have an incentive to
reduce the user's allocation.

To preclude such manipulations, we propose a combination of
lightweight cryptographic methods and incentive engineering. In
particular, we take advantage of the presence of {\em multiple}
sellers in envisioned applications. Precisely, we need enough sellers
that share the same reserve price.
 
The approach, \mechname{}, extends the \bks{} scheme described so far,
using a trusted central server for accounting and validation, logically
inserting it between the buyer and seller. Buyers pay the center
directly, and the center pools the revenue across sellers in a
particular way, making payments to sellers that address incentive
concerns.

We do not assume that the central server can observe or enforce
anything about the sellers' routing policy.  We merely rely on it to
verify the cryptographic signatures on buyer bids to ensure that
sellers cannot tamper with bids, to keep track of account balances for
buyers and sellers, and to verify that bid perturbation is done
correctly. We explain the approach in the next section.

\subsection{The Align-Trust mechanism}

\bks{} specifies that sellers should route packets using a
highest-resampled-bid-first policy.  One way to think about the problem with
seller incentives in \bks{} is that prioritizing in favor of buyers with
resampled bids has negative value to sellers.  (Recall that a bid is only
resampled with some probability $\mu$.)  We fix this by first paying sellers the
perturbed bid value for data, rather than the rebate-adjusted, perturbed bid
amount.  From the seller's perspective, this makes the auction have ``first
price'' semantics, and aligns incentives. 

Because the \bks{} mechanism also pays rebates to buyers, just paying
the sellers based on perturbed bid amounts would leave the center with
a deficit. Instead, we compensate by taxing the sellers a percentage
of their revenue.  We ensure that sellers can only reduce their
individual tax rate by improving efficiency, so sellers maximize revenue by routing to
maximize efficiency (i.e., total value) with respect to the resampled
bids, which we establish is sufficient to retain incentives on the
buy-side of the market.

For the purpose of \mechname{}, we assume time is divided into \emph{accounting
periods}, perhaps a month long in practice. These periods encapsulate many
\bks{} style auctions, as various sellers provide bandwidth to buyers.

Recall that \bks{} allows a seller to adopt a reserve price to increase revenue.
To align trust on the sell-side, we insist that each seller selects a reserve
price  from a small set of reserve prices.

Based on this,  we can now consider the {\em pool} of sellers that select the
same reserve price. We apply the following system-wide payment mechanism for the
auctions involving these sellers:

\begin{definition}[\mechname{}]

Consider an accounting period, and a set of sellers $M$ with the same reserve
price.  Let $\Gamma(s)$ be the set of buyers served by seller $s$ during the
accounting period.

\begin{enumerate}
\item Charge each buyer the \bks{} payment $b_i x_i - R_i$ for each completed auction. Credit each seller $s \in M$ the {\em first-price} revenue at the perturbed bids, without including the rebates: $\sum_{i \in \Gamma(s)} x_i \tilde{b}_i$.

\item Randomly split the sellers in $M$ into two disjoint sets $S_1, S_2$ of
    equal size.  Let 
\[C_1 = \sum_{s \in S_1} \sum_{i \in \Gamma(s)} x_i (\tilde{b}_i - r)\]
denote the total credit above reserve price to sellers in $S_1$. Let 
\[T_1 = \sum_{s \in S_1} \sum_{i \in \Gamma(s)} (b_i - r) x_i - R_i\]
denote the total above-reserve payments from buyers associated with $S_1$.
Define $C_2$ and $T_2$ similarly for sellers in $S_2$.  Because the payments
from buyers include the \bks{} rebates, and the credits to sellers are at the
first-price resampled bids $\tilde{b}$, this leaves the center with a
\emph{deficit} $\delta = C-T$ for each set.  

\item To make up the deficit, the center will tax the sellers.  Define the {\em tax rate} for each set as:
$\mathit{tax}_1=\frac{\delta_2}{C_1},
\quad
\mathit{tax}_2=\frac{\delta_1}{C_2}.$
Collect $\mathit{tax}_1 C_1$ from
sellers in $S_1$, charging tax rate $\mathit{tax}_1$ uniformly across
all sellers.
Collect $\mathit{tax}_2 C_2$ from
sellers in $S_2$,  charging tax rate $\mathit{tax}_2$ uniformly across
all sellers.
\end{enumerate}
\end{definition}

This construction has several nice properties. First of all, 
as long as sellers continue to use \spq{} 
it has no
effect on buyer
incentives, because
\mechname{} is identical to \bks{} from the buyer's point of view.
In addition, we have:
\begin{lemma}
The total payment in \mechname{} is exactly balanced.
\end{lemma}
\begin{proof}
The reserve price payments effectively go directly from buyers to sellers.  The total above-reserve payment to the system is 
$T_1 - C_1 + T_2 - C_2 + \mathit{tax}_1 C_1 + \mathit{tax}_2 C_2 = 
0$
\end{proof}

A tax rate is admissible if it is no greater than one. This will be true with
high probability when no seller accounts for a large fraction of rebates, as
shown in the next lemma.

\begin{lemma} 
The probability that the tax rate to a seller is greater than 1 falls
exponentially quickly in the number of sellers, when the payments made to each
seller in step 1 are independently distributed according to one of a finite set
of distributions with bounded support.\footnote{These distributions correspond
to the types of situations in which sellers operate: occasionally sharing access
in cafes, running a permanent hotspot in an area without wi-fi infrastructure,
frequent sharing at a conference facility, etc.  It is reasonable to divide
sellers into a finite set of such environment types, with random variability
within each.}  We call this assumption about the payment distribution the
\emph{admissibility condition}.

\end{lemma}

\begin{proof}    
Suppose for convenience that the pool contains $2m$ sellers,
and let $X_1,\ldots, X_m$ denote a random variable for the above-reserve credit
made in step 1 to each of the sellers in $S_1$. Similarly, let
$Y_1,\ldots,Y_m$ denote a random variable for the total above-reserve credit made
in step 1 to each of the sellers in $S_2$.

Let's consider
$\mathit{tax}_1$. For this to be bounded above by 1, we need
\[C_2-T_2 \leq C_1\]
Dividing through by $m$, and writing $C_2/m= \overline{Y}$ 
(where $\overline{Y}$ is the empirical mean),
$C_1/m=\overline{X}$ 
(where $\overline{X}$ is the empirical mean),
and with $T_2/m = z > 0$,
we want to bound the probability
$\Pr(\overline{Y}-\overline{X}\leq z)$.
For this, it is sufficient to bound the probability that
$$\overline{Y}-E[Y]\leq \frac{z}{2} \text{ and } E[X]-\overline{X}\leq \frac{z}{2},$$ 
since in this case we have $\overline{Y}-\overline{X} \leq z,$ since $E[X]=E[Y]$ by assumption. 
The probability $\Pr(\overline{Y}-E[Y]>\frac{z}{2})$, falls exponentially
quickly in $z$ and $m$ by Hoeffding's inequality, and similarly for
$\Pr(E[X]-\overline{X}>\frac{z}{2})$. This completes the proof.
\end{proof}

We assume that the admissibility condition holds for the theoretical results that follow.

In addition, from a seller's viewpoint, \mechname{} transforms the mechanism
into a first-price auction (with respect to perturbed bids) with a tax collected
on revenue, where the seller cannot usefully manipulate his tax rate:

\begin{lemma}
\label{lem:align}
If the tax rate for a seller is weakly less than one, the seller's
strict preference is to allocate in order to maximize the total
value given perturbed bids.
\end{lemma}
\begin{proof}
Consider a seller $s$ in $S_1$. The tax rate for $s$ is $\frac{\delta_2}{C_1}$.
$\delta_2$ depends only on sellers in $S_2$, and $C_1$ is the sum of the
revenues of other sellers in $S_1$, which $s$ cannot affect, and the total
revenue for $s$.  Increasing revenue lowers the tax rate facing $s$, so as long
as the tax rate is less than 1, this is doubly good for $s$.\footnote{In
practice, it would make sense to cap the tax rate at 1, putting the risk on the
center instead of the sellers.}
\end{proof}

\subsection{Seller Incentives}

Under \mechname{}, seller incentives are aligned with using a routing policy
that maximizes total pre-tax revenue $\sum_{i} \tilde{b}_i x_i$. This follows
from Lemma~\ref{lem:align}. In this
section, we examine the effect that this has on whether sellers want to follow
\spq{}, and on the effect of potential deviations on buy-side incentives.  

First, we characterize situations where sellers cannot profit by deviating from
\spq{}.  Second, we look at situations when the seller may profit from using a
different routing policy, and show that such deviations improve efficiency with
respect to perturbed bids and preserve incentive alignment with truthful bidding
for buyers. The crucial property that we need to retain under seller deviations
is that of monotonicity, or a relaxed form of expected monotonicity.

\subsubsection{Demand models where \spq{} is optimal}

We first consider cases where the seller's selfish preference is to follow
\spq{}. The intuition is that deviating results in an immediate drop in revenue
from sending lower priority traffic ahead of higher priority traffic, so in
order for this to increase revenue, the seller has to expect to make up the lost
revenue later. If this is not possible under a given demand model then
\spq{}is optimal for the seller. In the following, we consider  {\em ex post
efficiency}, which requires that an allocation rule maximizes total realized
value whatever the bids and whatever the realized demand. The following lemma
follows from Lemma~\ref{lem:align}, given admissibility.

\begin{lemma}
\label{lemma:efficient-nonmanip}
If \spq{} is {\em ex post} efficient (with respect to perturbed bids), 
the seller cannot increase his revenue by deviating from \spq{}.
\end{lemma}

We now specialize the natural demand models to understand when
\spq{} will be ex post efficient. A sufficient property is that the
demand model be memoryless, so that it does not depend on the routing
policy, with demand invariant to the total allocation made so far:
$d_i(t,x)=d_i(t,y)$ for all $x,y$ and for all realizations of a user's
random demand model. \spq{} is greedy in regard to bid, and if there
is no impact on future demand from a deviation from myopic value
maximization in the current period then \spq{} will maximize realized
value, so applying Lemma~\ref{lemma:efficient-nonmanip}
gives the following:
\begin{theorem}
If the user demand models are memoryless, then the seller maximizes revenue
by using \spq{}, and the \bks{} mechanism combined with \mechname{} 
retains truthfulness-in-expectation for buyers despite
seller self-interest.
\end{theorem}
Note that memoryless demand models are natural in the sense of our earlier definition.

\subsubsection{Other Demand Models}

The seller can benefit by deviating from \spq{} for some natural demand models.
For example, consider the model of an impatient buyer, as in the example in
Section~\ref{sec:natural-examples}. Such a buyer will stop transmitting if it
does not get some minimum amount of traffic by some time.  If the seller knows
this, it can sometimes increase revenue by increasing the buyer's priority,
increasing her allocation over the minimum, and ensuring continuing transmission
and continuing revenue.

We first consider the extreme case in which the seller can form a perfect
prediction of the future demand and bid values.
Let $\Delta$ denote a set of routing policies considered by
the seller, and let $x_i(\rho)$ be the total allocation to $i$ under
routing policy $\rho$, given the realized demand of all buyers.
\begin{lemma}
\label{lem:opt-is-monotone}
If the seller implements the optimal routing policy in $\Delta$, maximizing
\[\argmax{\rho \in \Delta} \sum_{i} \tilde{b}_i x_i(\rho),\]
with respect to perturbed bids $\tilde{b}$, then the implied allocation rule is {\em ex post} monotone and
the resulting mechanism is truthful-in-expectation for buyers.
\end{lemma}
\begin{proof}
The optimal allocation rule is {\em ex post} monotone. To see this,
consider the optimal allocation $x=(x_1,\ldots,x_n)$, denoting
the total amount of capacity used by each buyer by her departure,
based on perturbed bids $\tilde{b}=(\tilde{b}_1,\ldots,\tilde{b}_n)$.
Consider an alternate (perturbed) 
value $\tilde{b}'_i>\tilde{b}_i$, and new optimal
allocation $x'$. Since both $x$ and $x'$ are available
at bids $\tilde{b}$ and $(\tilde{b}'_i,\tilde{b}_{-i})$, 
then by optimality we have:
\begin{align*}
\tilde{b}_ix_i+\tilde{B}_{-i}(x_{-i})&\geq \tilde{b}_i x_i'+\tilde{B}_{-i}(x_{-i}')\\
\tilde{b}_i'x_i'+\tilde{B}_{-i}(x_{-i}')&\geq \tilde{b}_i'x_i+\tilde{B}_{-i}(x_{-i}),
\end{align*}
where $\tilde{B}_{-i}(x_{-i})$ is the total (perturbed)
bid for the allocation $x_{-i}$
to all buyers except i.
Adding and collecting terms, we have
\begin{align*}
\tilde{b}'_i(x'_i-x_i)\geq \tilde{b}_i(x'_i-x_i),
\end{align*}
and since $\tilde{b}'_i>\tilde{b}_i$ we need $x'_i\geq x_i$. This is the condition
for {\em ex post} monotonicity, and completes the proof.
\end{proof}

This illustrates the basic way in which the use of \mechname{} aligns
incentives for sellers so that they act in a way that retains buy-side
truthfulness.

The assumption that the seller is computing the {\em ex
 post} optimal allocation with respect to realized demand and
realized bids is very strong, and can be relaxed. 
Rather than ex post monotone, a weaker property of {\em monotone-in-expectation}
requires the expected total amount of bandwidth used by a buyer to be weakly
increasing with bid value. The expectation is taken with respect to a
probabilistic model of demand and bid values and also considering the random
perturbation of bids in \bks{}.

The dominant-strategy equilibrium property for buyers no longer holds for this
seller deviation, because the optimality of the seller's policy holds only in
expectation, given distributional assumptions about other buyers. In its place,
we adopt the standard notion of {\em Bayes-Nash incentive compatibility}, which
in our case requires that truthful bidding and following the greedy policy
maximizes a buyer's expected utility, given that other buyers do the same.

\begin{theorem}
\label{thm:estimate-is-monotone}
If the seller implements a routing policy that is optimal in expectation,
solving
\[\argmax{\rho \in \Delta} \sum_{i} E\!\left[\tilde{b}_i x_{i,\geq t}(\rho)\right],\]
forward from any time $t$, where $x_{i,\geq t}(\rho)$ is the total additional
allocation (bytes) to $i$ between $t$ and the departure time of $i$, 
 given perturbed bids $\tilde{b}$, then
the resulting mechanism is Bayes-Nash incentive compatible.
\end{theorem}
\begin{proof}
Consider buyer $i$ and arrival period $t=\alpha_i$. By the expected
optimality of the routing policy,
(where 
the expectation is taken with respect to the probabilistic
demand model of buyers and a distribution on perturbed bid values,
itself induced by the \bks{} randomization and a distribution
on buyer values), we have
\begin{align*}
\tilde{b}_i z_i + \tilde{B}_{-i}(z_{-i,\geq t}) &\geq \tilde{b}_i z'_i + \tilde{B}_{-i}(z'_{-i,\geq t})\\
\tilde{b}'_i z'_i+\tilde{B}_{-i}(z'_{-i,\geq t})&\geq \tilde{b}_i'z_i + \tilde{B}_{-i}(z_{-i,\geq t}),
\end{align*}
where $z_i$ is the expected allocation to buyer $i$ forward
from $t$ given perturbed bid  $\tilde{b}_i$, $z_{-i,\geq t}$ the expected
allocation to other buyers forward from $t$, $\tilde{B}_{-i}(\cdot)$ denotes
the expected value from this allocation given realized, perturbed bids
reported so far and the distribution on future (perturbed) bids, 
$\tilde{b}'_i>\tilde{b}_i$ is an alternate bid, and $z'_i$ and $z'_{-i,\geq t}$ 
denote the respective allocation quantities under this
alternate bid. 

Proceeding in the same way as Lemma~\ref{lem:opt-is-monotone},
we can add and collect terms and obtain
\begin{align*}
\tilde{b}'_i(z'_i-z_i)&\geq \tilde{b}_i(z'_i-z_i),
\end{align*}
and conclude since $\tilde{b}'_i>\tilde{b}_i$ that
 $z'_i\geq z_i$. This is the condition
for monotone-in-expectation. From this, it then follows from the standard
\bks{} incentive arguments (that in turn rely on monotonicity)
that the mechanism aligns incentives for buyers,
with Bayes-Nash equilibrium adopted in place of truthful-in-expectation 
since monotonicity relies
on the other buyers following
the equilibrium, so that the seller's probabilistic model about
the world is correct in regard to demand and bids.
\end{proof}

The weakening from truthfulness-in-expectation to Bayes-Nash
equilibrium arises because the seller's policy is only
optimal-in-expectation, and this in turn relies on the seller having a
correct probabilistic model of buyer demand--- and thus equilibrium behavior by
buyers.

{\bf Example 3.} 
{\em An example to illustrate this idea: imagine a case
where the seller knows that a buyer is impatient, with constant demand
10 in each round for the first 60 rounds, leaving after that if they
do not get a minimum amount $m$, drawn uniformly from $[300,450]$.
There is no uncertainty about this buyer's demand, and the seller
computes his best estimates of the other buyers' future demand.  Using
these estimates, the seller can compute the probability that this
buyer will reach each amount between 300 and 450 under \spq{} routing.
Then it can consider switching the routing order in some rounds to
increase the buyer's allocation, and consider how much that increases
the likelihood of the buyer making it to her (unknown) $m$ and the
likely future gains from that, versus the immediate revenue reduction
such a deviation requires.  The intuition for why this is still
monotone should be clear--- if the buyer's traffic is worth more,
there is a smaller revenue reduction from increasing her priority, and a
larger potential future payoff, so the seller would be more likely to
increase the buyer's allocation.}
\smallskip

Buyers do not need to
know whether the seller is using \spq{} or an expected-efficiency
maximizing policy, since \bks{} remains monotone, and thus incentives
are aligned with truthful bidding either way.  The buyer only needs to
believe that the seller is not irrationally reducing his expected
revenue by changing priorities away from \spq{}. 
Another useful property is that these results hold for any set of
routing policies $\Delta$ considered by the router--- as long as the
router always chooses the best policy (in expectation) from a fixed
set, we still get monotonicity, and thus buyer truthfulness.

\section{Simulation Results}

We now present simulations that
explore the quantitative difference between mechanisms, and illustrate
the effects predicted by our theoretical analysis.  We confirm that \bks{} achieves
arbitrarily close approximations to allocative efficiency for simple
demand models, show that reserve prices can increase efficiency of
non-prioritized routing methods, demonstrate a scenario where sellers
have practical efficiency-improving deviations, and examine the
distributional effects of revenue pooling.

\subsection{Our simulator}

We adopt a custom flow-level simulator, closely matching our theoretical model.
Each epoch is one second, and we run simulations for 10 minutes, or 600 epochs.
Demand and capacity are specified in KBps.  The simulator supports arbitrary
user demand models.  We use a stochastic demand model motivated by patterns
observed in real networks, somewhat simplifying the trace-based model described
in~\cite{papadopouli05demand}. A user's demand is composed of \emph{flows} with
total demand equal to the sum of the demands of the flows active at any time,
flow arrival times are a Poisson process, flow duration is sampled from a
lognormal distribution, and each flow has Poisson demand. 

We use the following parameters: flow durations have both mean and standard
deviation 30 seconds, average flow inter-arrival times are 30 seconds, and each
flow has average rate of either 10 or 30 KBps, as specified. This results in
average demand equal to the average flow rate.

\begin{figure}[t]
\begin{center}
\includegraphics[width=0.65\linewidth]{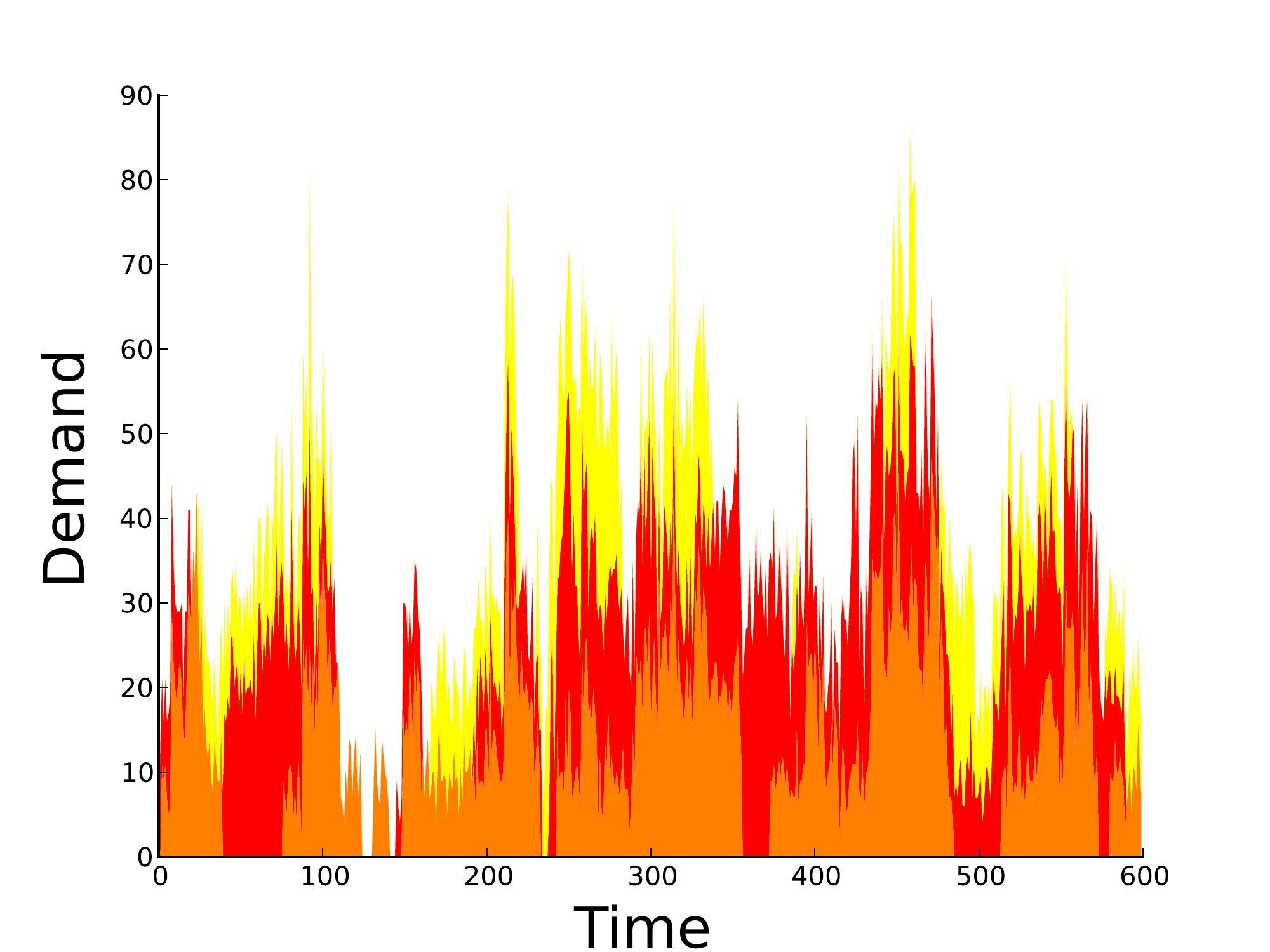}
\caption{Sample demand realization for 3 buyers with our trace-inspired model
for 600 simulated secs.  Each color is a buyer's demand. 
} 
\label{fig:demand-trace} 
\end{center}
\end{figure}

Fig.~\ref{fig:demand-trace} shows an example demand trace of three such buyers, each with demand 10 KBps.  The pattern is bursty both at short timescales, reflecting the Poisson demand of each flow, and at longer timescales, reflecting the random flow arrival process.

The simulator supports the \spq{}, \fq{}, and \fifo{} routing policies described
in Section~\ref{sec:policies}.  Unless otherwise specified, we simulate \vmm{}
with \spq{} routing, \bks{} with \spq{} routing with priorities based on
resampled bids, and fixed price with \fq{} and \fifo{} routing.  The \bks{}
resampling frequency $\mu$ is 0.2, chosen as a compromise between reducing the
magnitude of rebates and not sacrificing efficiency.  The results are not very
sensitive to the choice of $\mu$.  We show the expected allocation over 1000
allocations for \bks{}.  

We simulate truthful bidding for both \bks{} and \vmm{}, even though
\vmm{} is not a truthful mechanism.  For memoryless demand models,
this means that \vmm{} results in the optimal allocation, acting as a
benchmark for the truthful mechanisms.

\subsection{Efficiency}

From our analysis, we know that \spq{} is efficient for
memoryless demand models, but may not be efficient for other demand
models.  Here, we look at both cases, and look at the magnitude of
this effect in some scenarios.  We also examine the drop in efficiency
due to bid resampling in \bks{}.

\begin{figure}[t]
\begin{center}
\includegraphics[width=0.65\linewidth]{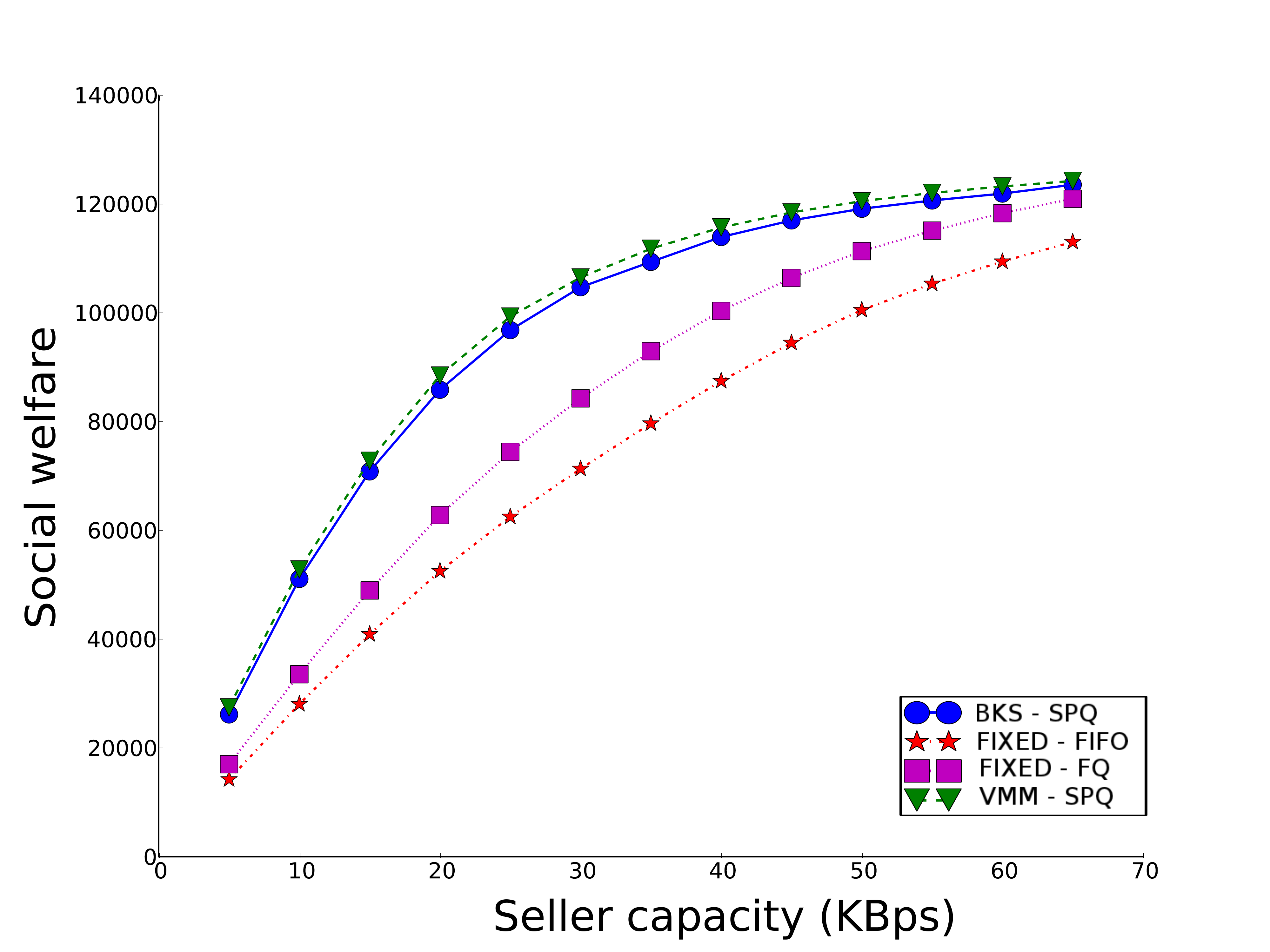}
\caption{Efficiency of routing policies with truthful bidding. The \vmm{} allocation is optimal in this setting, so the drop to \bks{} is the efficiency loss due to bid resampling necessary to obtain truthfulness.  \spq{} routing for \vmm{} and \bks{} is much more efficient than the non-prioritized routing policies.}
\label{fig:efficiency} 
\end{center}
\end{figure}

Fig.~\ref{fig:efficiency} shows the social welfare (sum of buyer and seller
utilities) for a scenario with three buyers.  The first two have per-KB values
10 and 4, and average demand 10 KBps each, simulating a high value and a medium
value user.  The last buyer has a per-KB value of only 1, but has a higher
average demand--30 KBps.   All buyers use the stochastic demand model described
above, and we compare  four different mechanisms--\bks{} and \vmm{}, both with
\spq{} routing and truthful bids, and \fq{} and \fifo{}, with a fixed price of
1.  (We study the effect of changing this fixed price in the next experiment).
As the seller capacity increases, the social welfare increases for all the
mechanisms, until all the demand is satisfied.  The demand models in this
simulation are memoryless, so \spq{} routing for \vmm{} is
optimal given that we insist on truthful bidding.  

The small difference between the performance of \vmm{} and \bks{} is
the efficiency cost of the sampling necessary for the
truthfulness provided by \bks{}.  \spq{} routing for \bks{} is much
more efficient than either of the non-prioritized fixed price
policies.  The difference between \fq{} and \fifo{} is also
interesting--with \fifo{}, the allocation is proportional to demand,
so the low value, high-demand buyer gets three-fifth of the channel on
average, reducing efficiency.

\begin{figure}[t]
\begin{center}
\includegraphics[width=0.65\linewidth]{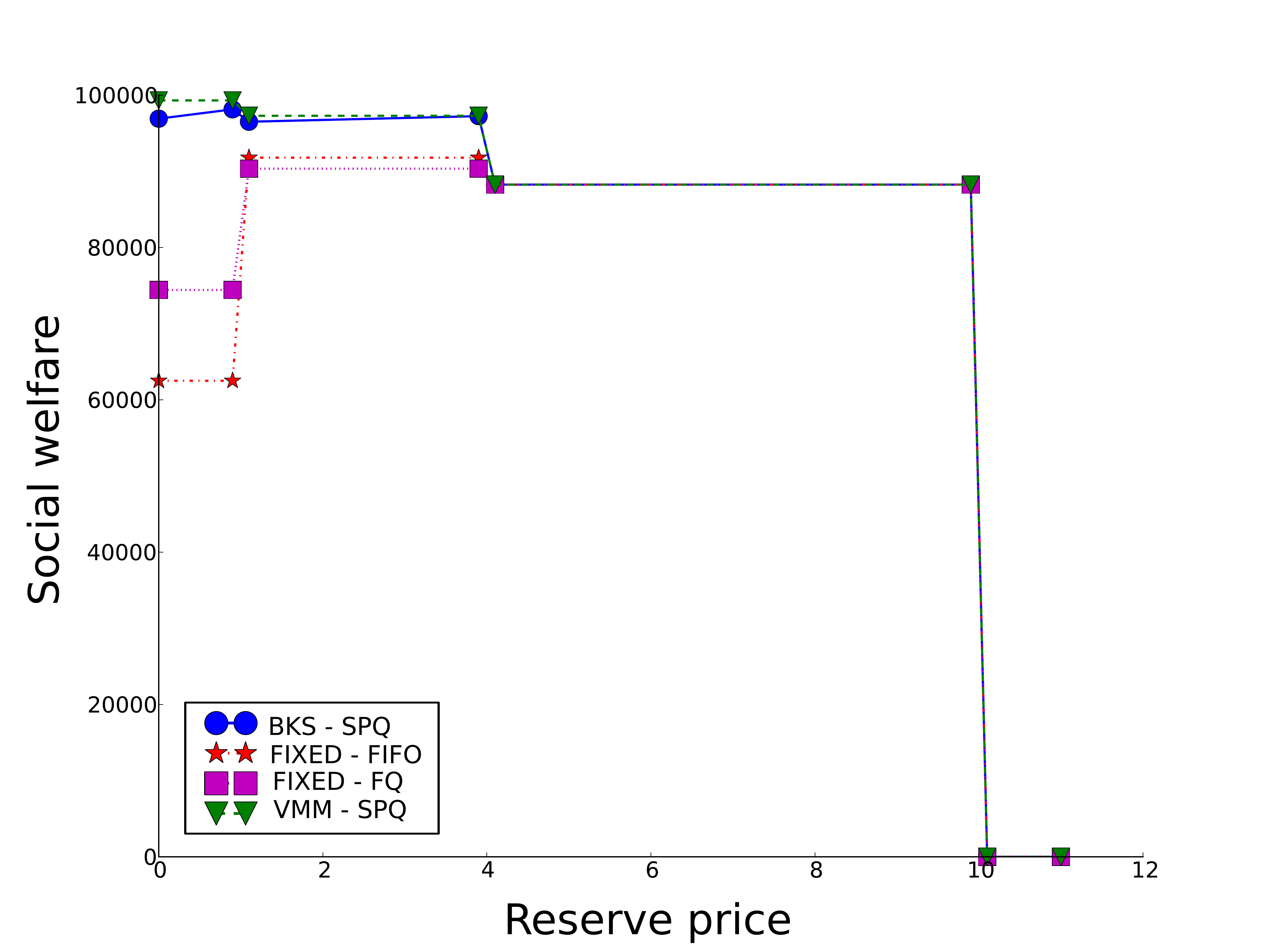}
\caption{Efficiency of routing policies as reserve prices varies.  Seller capacity is 25 KBps. Increasing the reserve price improves the efficiency of \fq{} and \fifo{}.}
\label{fig:efficiency-vs-reserve} 
\end{center}
\end{figure}

Fig.~\ref{fig:efficiency-vs-reserve} shows the effect of reserve price on
efficiency, in the same scenario, with seller capacity fixed at 25 KBps.
Raising the reserve price slightly to avoid allocating to the lowest value buyer
improves efficiency for \fifo{} and \fq{}.  Since \vmm{} and \bks{} naturally
prioritize high value data, raising the reserve price removes opportunities to

\begin{figure}[t]
\begin{center}
\includegraphics[width=0.65\linewidth]{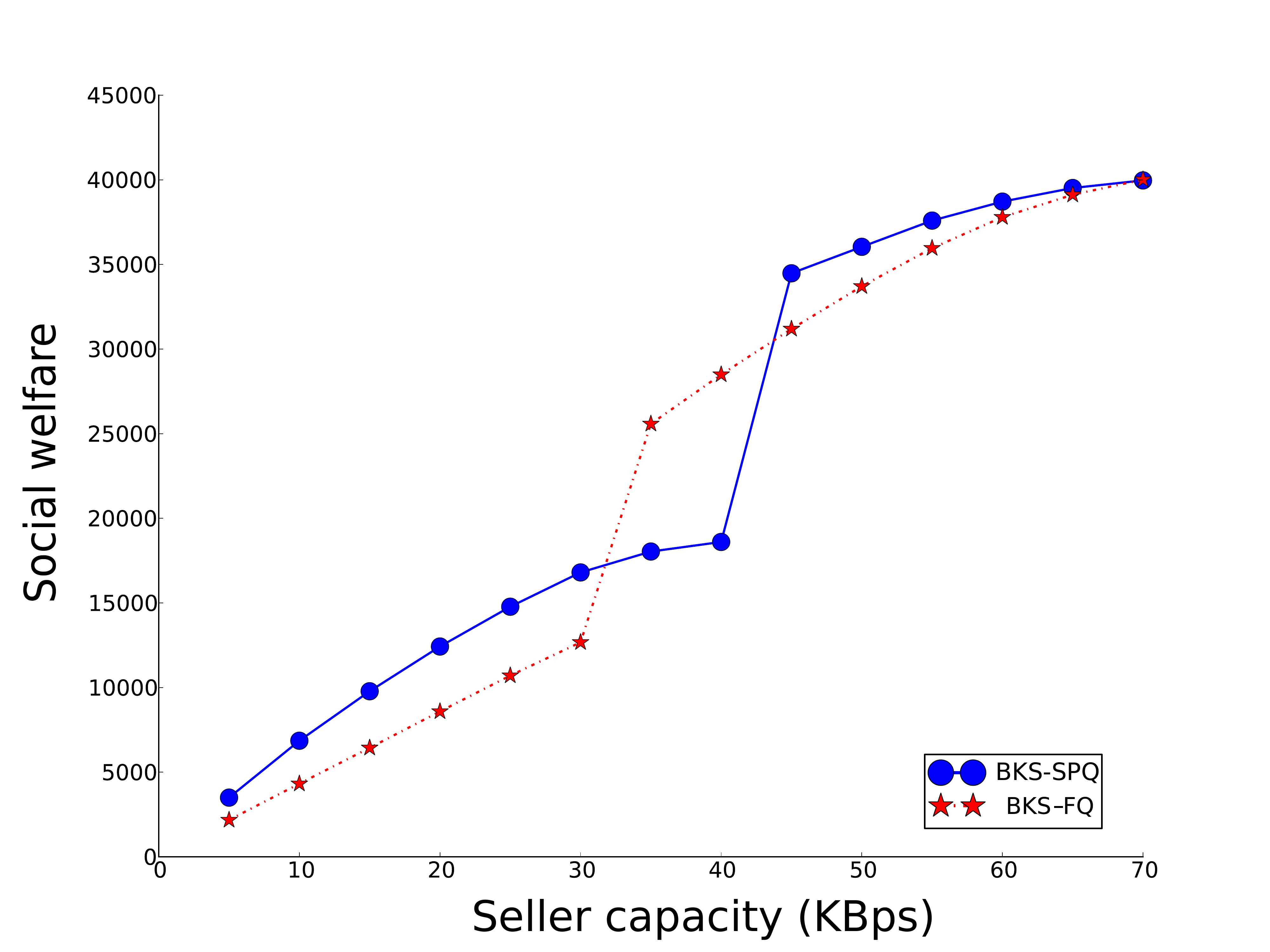}
\caption{Efficiency for \spq{} and \fq{} when high value buyers leave early, and
a low value buyer is impatient.  Deviating to \fq{} from \spq{} is profitable
for a range of seller capacities.}
\label{fig:efficiency-impatient} 
\end{center}
\end{figure}

We now look at a case where a seller can profit by deviating from \spq{}.  The
scenario is similar to the one above: three buyers, with per-KB values 10, 4,
and 1, and average demands 10, 10, and 30 KBps.  However, now the high and
medium value buyers depart after 90 seconds, and the low value buyer is
impatient, leaving after 60 seconds if her allocation is less than 500KB, and
otherwise continuing to send for the full 600 seconds.  If the seller fails to
allocate at least 500KB to the third buyer in the first minute, she misses out on
forwarding extra traffic later, once the two higher value bidders are gone.
Fig.~\ref{fig:efficiency-impatient} shows the efficiency of \bks{} with \spq{}
and \fq{} routing.  For some intermediate capacities, the value-ignorant \fq{}
routing policy is better, because it results in the low priority buyer
continuing to send.  Of course, if the seller knew the impatience
threshold, it could use a better hybrid strategy, allocating just enough
to buyer 3 to ensure she stays, and then allocating optimally to others.

\subsection{Revenue pooling}

Next, we study the effects of revenue pooling in \mechname{}, first
looking at a simplified case where all sellers have the same buyer
distribution, and then a more complex case where there are several
seller distributions.  We confirm our theoretical results, showing
that the center will not have to run a deficit, and also show that
pooling reduces the variance of seller revenue.
Fig.~\ref{fig:revenue-pooling-same} shows the pooled revenue vs
unpooled revenue for 200 sellers from a single distribution, with
seller capacity 40, buyers with per-KB values 2 and 3 and demand 10
KBps, a buyer with value 5 and demand 30 KBps, and no reserve price.
In addition to the primary goal of making \bks{} faithful for sellers,
pooling drastically reduces the variance in seller revenue, with
non-trivial pooled revenue for all sellers.  This also shows that the
tax rate for both pools is lower than 1, as predicted by the theory.

\begin{figure}[t]
\begin{center}
\includegraphics[width=0.65\linewidth]{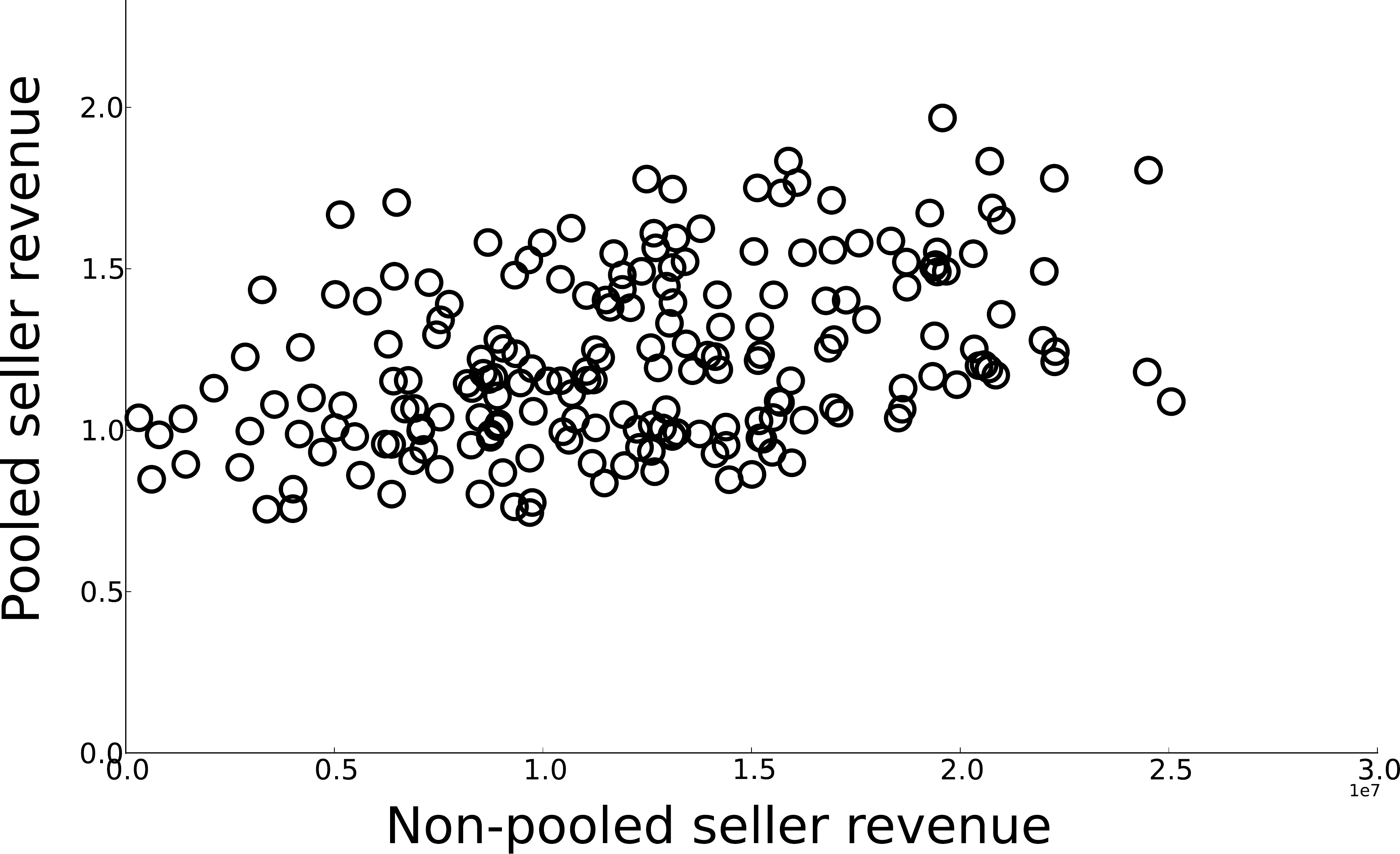}
\caption{Effect of revenue pooling when sellers are similar.}
\label{fig:revenue-pooling-same} 
\end{center}
\end{figure}

\begin{figure}[t]
\begin{center}
\includegraphics[width=0.65\linewidth]{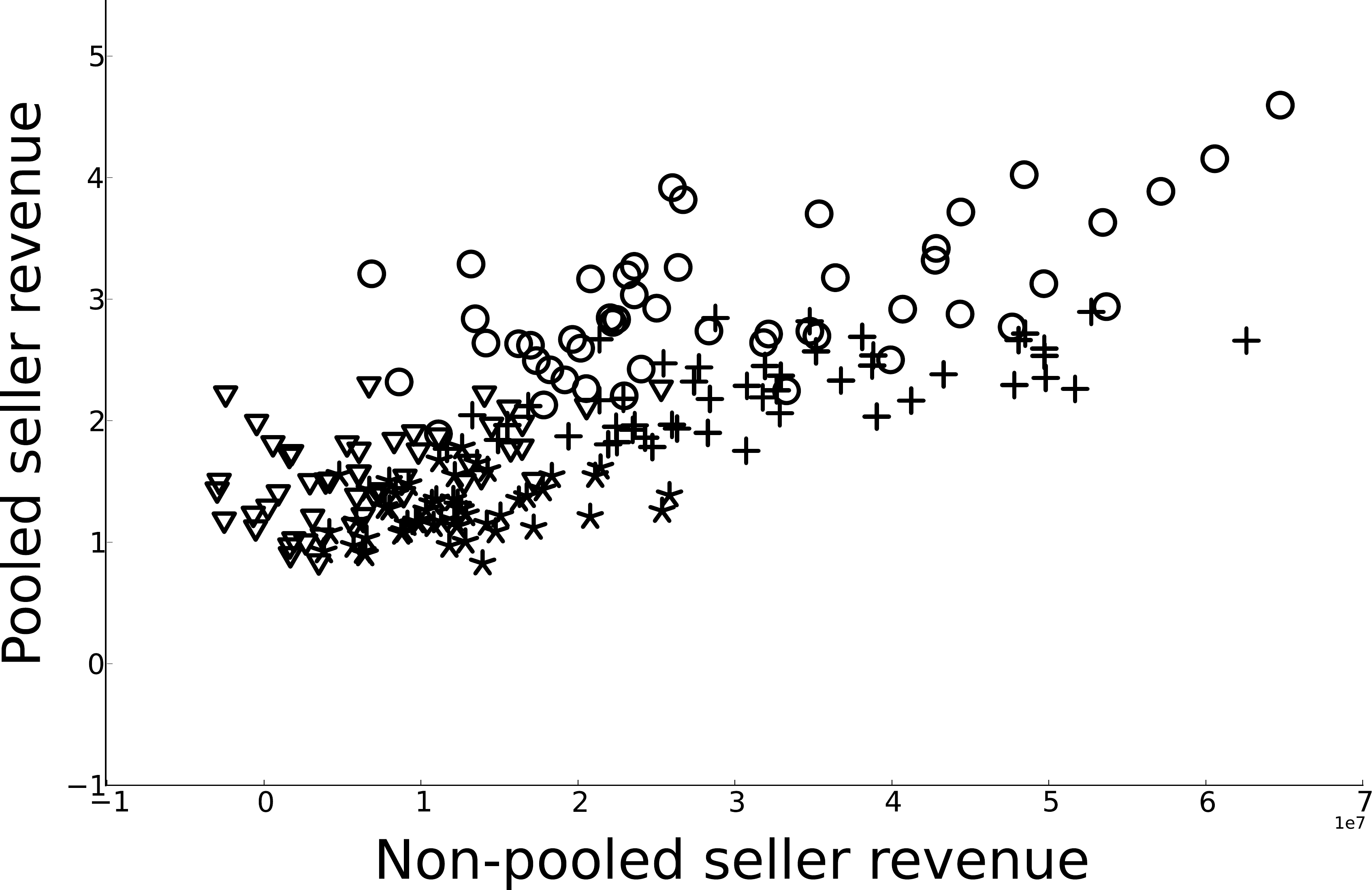}
\caption{Effect of revenue pooling when sellers are different.}
\label{fig:revenue-pooling-varied} 
\end{center}
\end{figure}

Fig.~\ref{fig:revenue-pooling-varied} shows the pooled revenue vs unpooled
revenue for 200 sellers from four different settings: seller capacities 40 or
60, and buyers with total demand 50 or 70, with higher buyer value when demand
is 70. These are indicated by different markers.  Note the dramatic reduction in
variance of pooled seller revenue for each seller category: without revenue
pooling, some sellers have negative revenue while others have very high revenue.
Pooling reduces this variance, while preserving the relative ordering between
seller types--sellers with higher non-pooled revenue have higher expected pooled
revenue.

\section{System implementation issues}
\label{sec:implementation}
Our scheme has three main components: a protocol for buyers to send bids to the
sellers, a priority queueing facility for sellers to allocate bandwidth to buyer
flows according to BKS, and a central server for revenue pooling and accounting.
We have prototyped the first two components.

An app on the buyer's side sends bids (along with duration of
validity of the bid) to a service on the seller device. The seller
device resamples the bids and uses them for
allocating bandwidth, using the standard traffic control (\texttt{tc})
facilities in the Linux kernel (also used by Android).
Specifically, upon receiving bids or updates to bids, the seller device sets
(\texttt{tc}) filters to classify flows into one of a fixed number of priority
classes based on their resampled bids, which enables the kernel to enqueue
packets in the appropriate priority queue.  Thus, we get efficient strict
priority queuing without any kernel modifications, with the caveat that we can only use a small
fixed number of queues.  This implies that we need to round down the resampled
bids to the nearest priority level. However, this modification preserves
truthfulness since the modified allocation function is still (weakly) monotone
in bids. Payments are made according to the un-rounded resampled bids.

\section{Conclusion}

We have introduced an approach to prioritized bandwidth access in a dynamic
environment. The method succeeds in aligning incentives on both the buy-side and
the sell-side of the market for a variety of stochastic demand models.  While
auctions can introduce an extra burden on users over charging a flat fee, they
can provide for efficient allocation and the complexity can be hidden through
automated bidding agents~\cite{chi10Seuken, Seuken_ec12}.

There are many areas for future work.  Natural extensions to the theory include
expanding to more expressive demand models, and considering valuation models
other than linear per byte such as value per rate and multi-dimensional
valuations. The \bks{} scheme has been extended to handle the appropriate
generalization of monotonicity through a self-sampling
approach~\cite{babaioff13}, and the challenge would be design dynamic
prioritization schemes that satisfy this cyclic monotonicity property for
stochastic demand models.

We make several assumptions that are useful in analysis and developing
engineering insights, even if they might not quite hold in reality.  More extensive
and detailed packet level simulations would be useful to study situations where
our assumptions such as natural demand models and linear per-byte valuations do
not hold. Analyzing real demand to determine the extent to which it confirms to
our models is also future work. Finally, developing our dynamic
prioritization system prototype further by designing appropriate user interfaces for
eliciting buyer values, integrating with applications to allow setting
traffic values, and evaluating performance such as computational overhead and
energy consumption also remains to be done.

\noindent \textit{Acknowledgements:} Research was sponsored by the Army Research Laboratory and was accomplished
under Cooperative Agreement Number W911NF-09-2-0053. 


\end{document}